\documentclass{llncs}

\usepackage{amsmath,amssymb}
\usepackage{enumitem}
\usepackage[dvips]{epsfig}
\usepackage{epsf}
\usepackage{float}
\usepackage{array}
\usepackage{vaucanson-g}

\newcolumntype{M}{>{$}l<{$}}
\newcolumntype{N}{>{$}r<{$}}

\DeclareMathOperator{\alphabet}{alph}

\begin{document}

\frontmatter

\title{Finite Orbits of Language Operations}

\author{\'Emilie Charlier\inst{1} \and
Mike Domaratzki\inst{2} \and 
Tero Harju\inst{3} 
\and Jeffrey Shallit\inst{1}}

\institute{University of Waterloo,
Waterloo, ON  N2L 3G1 Canada 
\email{echarlier@uwaterloo.ca, shallit@cs.uwaterloo.ca}
\and
University of Manitoba,
Winnipeg, MB  R3T 2N2 Canada
\email{mdomarat@cs.umanitoba.ca}
\and
University of Turku, 
FIN-20014 Turku, Finland 
\email{harju@utu.fi} 
}

\maketitle

\def\pref{{\rm pref}}
\def\suff{{\rm suff}}
\def\fact{{\rm fact}}
\def\subw{{\rm subw}}

\begin{abstract}
We consider a set of natural operations on languages, and prove that
the orbit of any language $L$ under the monoid generated by
this set is finite and bounded, 
independently of $L$.  This generalizes previous results about
complement, Kleene closure, and positive closure.
\end{abstract}

\section{Introduction}

If $t, x, y, z$ are (possibly empty) words with $t = xyz$, we say 
\begin{itemize}
\item $x$ is a {\it prefix} of $t$;
\item $z$ is a {\it suffix} of $t$;
and
\item $y$ is a {\it factor} of $t$.
\end{itemize}
If $t = x_1 t_1 x_2 t_2 \cdots x_n t_n x_{n+1}$ for some $n\geq 1$ and
some (possibly empty) words
$t_i, x_j$, $1 \leq i \leq n$, $1 \leq j \leq n+1$, then $t_1 \cdots t_n$
is said to be a {\it subword} of $t$.  Thus a factor is a contiguous
block, while a subword can be ``scattered''.

Let $L$ be a language over the finite alphabet $\Sigma$, that is, 
$L \subseteq \Sigma^*$.  We consider the following eight
natural operations
applied to $L$:
\begin{eqnarray*}
k: && L \rightarrow L^* \\
e: && L \rightarrow L^+  \\
c: && L \rightarrow \overline{L} = \Sigma^* - L  \\
p: && L \rightarrow \pref(L) \\
s: && L \rightarrow \suff(L) \\
f: && L \rightarrow \fact(L) \\
w: && L \rightarrow \subw(L) \\
r: && L \rightarrow L^R.
\end{eqnarray*}

Here 
\begin{eqnarray*}
\pref(L) &=& 
\lbrace x \in \Sigma^* \ : \ x \text{ is a prefix of
some } y \in L \rbrace; \\
\suff(L) &=& \lbrace x \in \Sigma^* \ : \ x \text{ is a suffix of
some } y \in L \rbrace;  \\
\fact(L) &=& \lbrace x \in \Sigma^* \ : \ x \text{ is a factor of
some } y \in L \rbrace; \\
\subw(L) &=& \lbrace x \in \Sigma^* \ : \ x \text{ is a subword of
some } y \in L \rbrace; \\
L^R &=& \lbrace x \in \Sigma^* \ : \ x^R \in L \rbrace;
\end{eqnarray*}
where $x^R$ denotes the reverse of the word $x$.  

We compose these operations as follows:  if 
$x = a_1 a_2 \cdots a_n \in \lbrace k,e,c,p,s,f,w,r \rbrace^*$, then
$$x(L) = a_1(a_2(a_3(\cdots ( a_n (L)) \cdots ))).$$
Thus, for example, $ck(L) = \overline{L^*}$.  We also write $\epsilon(L) = L$.  

Given two elements $x, y \in \lbrace k,e,c,p,s,f,w,r \rbrace^*$, we 
write $x \equiv y$ if $x(L) = y(L)$ for all languages $L$, and we write
$x \subseteq y$ if $x(L) \subseteq y(L)$ for all languages $L$.

Given a subset $S \subseteq \lbrace k,e,c,p,s,f,w,r \rbrace$, we can
consider the orbit of languages
$$ {\cal O}_S (L) = \lbrace x(L) \ : \ x \in S^* \rbrace $$
under the monoid of operations generated by $S$.  We are interested
in the following questions:  when is this monoid finite?  Is the
cardinality of ${\cal O}_S (L)$ bounded, independently of $L$?

These questions were previously investigated for the sets
$S = \lbrace k,c \rbrace$ and $S = \lbrace e,c \rbrace$ \cite{Peleg,BGS}, where
the results can be viewed as the
formal language analogues of Kuratowski's celebrated ``14-theorem'' for
topological spaces
\cite{Kuratowski:1922,GJ}.
In this paper we consider the questions for other subsets of 
$\lbrace k,e,c,p,s,f,w,r \rbrace$.  Our main result is
Theorem~\ref{main} below, which shows finiteness for any subset of these
eight operations.

\section{Operations with infinite orbit}

We point out that the orbit of $L$ under an arbitrary operation need not be
finite.  For example, consider the operation $q$ defined by
$$ q(L) = 
\lbrace x \in \Sigma^* \ : \ x \text{ there exists $y \in L$ such that
$x$ is a proper prefix of $y$ } \rbrace .$$
Here by ``$x$ is a proper prefix of $y$'', we mean that $x$ is a prefix of $y$
with $|x| < |y|$.

Let $L = \lbrace a^n b^n \ : \ n \geq 1 \rbrace$.  Then it is
easy to see that the orbit 
$$ {\cal O}_{\lbrace q \rbrace}( L) = \lbrace L, q(L), q^2(L), q^3(L), \ldots
\rbrace$$
is infinite, since the shortest word in
$q^i(L) \ \cap \ a^+ b$ is $a^{i+1} b$.  

The situation is somewhat different if $L$ is regular:

\begin{theorem}
Let $q$ denote the proper prefix operation, and let $L$ be a regular
language accepted by a DFA of $n$ states.  
Then ${\cal O}_{\lbrace q \rbrace}( L) \leq n$,
and this bound is tight.
\end{theorem}

\begin{proof}
Let $M = (Q, \Sigma, \delta, q_0, F)$ be an
$n$-state DFA accepting $L$.  Note that a
DFA accepting $q(L)$ is given by $M' = (Q, \Sigma, \delta, q_0, F')$
where 
$$F' = \lbrace q \in Q \ : \ \text{ there exists a path of
length $\geq 1$ from $q$ to a state of $F$ } \rbrace.$$
Reinterpreting this in terms of the underlying transition diagram,
given a directed graph $G$ on $n$ vertices, and a distinguished set of
vertices $F$, we are interested in the number of different sets obtained
by iterating the operation that
maps $F$ to the set of all vertices that can reach a vertex in $F$
by a path of length $\geq 1$.  We claim this is at most $n$.  To see this,
note that if a vertex $v$ is part of any directed cycle,
then once $v$ is included,
further iterations will retain it.    Thus the number of distinct sets
is as long as the longest directed path that is not a cycle, plus $1$
for the inclusion of cycle vertices.  

To see that the bound is tight, consider the language 
$L_n = \lbrace \epsilon, a, a^2, \ldots, a^{n-2} \rbrace$, 
which is accepted by a (complete) unary DFA of $n$ states.  
Then $q(L_n) = L_{n-1}$, so this shows
$|{\cal O}_{\lbrace q \rbrace}( L_n)| = n$. \ \qed
\end{proof}

It is possible for the orbit under a single operation
to be infinite even if the operation is
(in the terminology of the next section) expanding and
inclusion-preserving.  As an example, consider the operation
of fractional exponentiation, defined by
\begin{displaymath}
n(L) = \lbrace x^\alpha \ : \alpha \geq 1 \text{ rational } \rbrace 
= \bigcup_{x \in L} x^+ p(\lbrace x \rbrace).
\end{displaymath}

\begin{proposition}
Let $L = \lbrace ab \rbrace$.
Then the orbit ${\cal O}_{\lbrace n \rbrace} (L)$ is infinite.
\end{proposition}

\begin{proof}
We have $aba^i \in n^i(\lbrace ab \rbrace)$, but
$a ba^i \not\in n^j(\lbrace ab \rbrace)$ for $j < i$.   \ \qed
\end{proof}

\section{Kuratowski identities}

Let $a: 2^{\Sigma^*} \rightarrow 2^{\Sigma^*}$ be an operation on
languages.  Suppose $a$ satisfies the following three properties:

\begin{enumerate}
\item $L$ is a subset of $a(L)$ (expanding);
\item If $L \subseteq M$ then $a(L) \subseteq a(M)$ (inclusion-preserving);
\item $a(a(L)) = a(L)$ (idempotent).
\end{enumerate}

Then we say $a$ is a {\it closure operation}. Examples of closure
operations include $k, e, p, s, f, $ and $w$.

Note that if $a,b$ are closure operations, then their composition $ab$ trivially satisfies properties
1 and 2 above, but may not satisfy property 3. 
For example, $pk$ is not idempotent, as can be seen by examining its action on
$L = \lbrace ab \rbrace$ ($aab \not\in pk(L)$, but $aab \in pkpk(L)$).


\begin{lemma}
Let $a \in \lbrace k,e \rbrace$ and $b \in \lbrace p,s,f,w \rbrace$.
Then $aba \equiv bab \equiv ab$.
\label{kpk}
\end{lemma}

\begin{proof}
We prove the result only for $b = p$; the other results are similar.

Since $L \subseteq a(L)$, we get $p(L) \subseteq pa(L)$, and
then $ap(L) \subseteq apa(L)$.
It remains to see $apa(L) \subseteq ap(L)$.

Any element of $a(L)$ is either $\epsilon$ or of the form
$t = t_1 t_2 \cdots t_n$ for some $n \geq 1$, where each $t_i \in L$.
Then any prefix of $t$ looks like 
$t_1 t_2 \cdots t_{i-1} p_i$ for some $i \geq 1$,
where $p_i$ is a prefix of $t_i$, and hence in $p(L)$.
But each $t_i$ is also in $p(L)$, so
this shows 
\begin{equation}
pa(L) \subseteq ap(L). \label{paap} 
\end{equation}
Since $a$ is a closure operation,
$apa(L) \subseteq aap(L) = ap(L)$.

Similarly, we have $ap(L) \subseteq pap(L)$. Substituting
$p(L)$ for $L$ in (\ref{paap}) gives
$pap(L) \subseteq app(L) = ap(L)$. \ \qed
\end{proof}


\begin{lemma}
The operations $kp,ks,kf,kw,ep,es,ef$ and $ew$ are closure
operations.
\label{co-lemma}
\end{lemma}

\begin{proof}
We prove the result for $kp$, with the other results being similar.
It suffices to prove property 3.  
From Lemma~\ref{kpk} we have $pkp(L) = kp(L)$.  Applying $k$ to both
sides, and using the idempotence of $k$,
we get $kpkp(L) = kkp(L) = kp(L)$. \ \qed
\end{proof}

If $a$ is a closure operation, and $c$ denotes complement,
then it is well-known (and shown, for
example, in \cite{Peleg}) that $acacaca \equiv aca$.  However, we will
need the following more general observation, which seems to be new:

\begin{theorem}
Let $x, y$ be closure operations.  Then $xcycxcy \equiv xcy$.
\label{kura}
\end{theorem}

\begin{proof}
$xcycxcy \subseteq xcy$:  We have $L \subseteq y(L)$ by the expanding
property.  Then $cy(L) \subseteq c(L)$.  By the inclusion-preserving
property we have $xcy(L) \subseteq xc(L)$.  Since this identity holds
for all $L$, it holds in particular for $cxcy(L)$.  Substituting, we
get $xcycxcy(L) \subseteq xccxcy(L)$.  But $xccxcy(L) = xcy(L)$ by
the idempotence of $x$.

$xcy \subseteq xcycxcy$:  We have $L \subseteq x(L)$ by the expanding
property.  Then, replacing $L$ by $cy(L)$, 
we get $cy \subseteq xcy$.  Applying $c$ to both sides, we get
$cxcy \subseteq ccy = y$.  Applying $y$ to both sides, and using
the inclusion-preserving property and idempotence, we get
$ycxcy \subseteq yy = y$.  Applying $c$ to both sides, we get
$cy \subseteq cycxcy$.  Finally, applying $x$ to both sides and using
the inclusion-preserving property, we get $xcy \subseteq xcycxcy$.
\ \qed
\end{proof}

\begin{remark}
Theorem~\ref{kura} would also hold if $c$ were replaced by any
inclusion-reversing operation satisfying $cc \equiv \epsilon$.
\end{remark}

As a corollary, we get \cite{Peleg,BGS}:

\begin{corollary}
If $S = \lbrace a, c \rbrace$, where $a$ is any closure operation,
and $L$ is any language,
the orbit ${\cal O}_{S} (L)$ contains at most $14$ distinct languages.
\end{corollary}

\begin{proof}
The $14$ languages are given by the image of $L$ under the 14 operations
$$\epsilon, a, c, ac, ca, aca, cac, acac, caca,
acaca, cacac, acacac, cacaca, cacacac .$$ 
\ \qed
\end{proof}

\begin{remark}
Theorem~\ref{kura}, together with Lemma~\ref{co-lemma}, thus gives 
196 separate identities.
\end{remark}

In a similar fashion,
we can obtain many kinds of Kuratowski-style identities involving
$k, e, c, p, s, f, w $ and $r$.    

\begin{theorem}
Let $a \in \lbrace k, e \rbrace$ and $b \in \lbrace p,s,f,w \rbrace$.
Then we have the following identities:
\begin{enumerate}[start=4]
\item $abcacaca \equiv abca$
\item $bcbcbcab \equiv bcab$
\item $abcbcabcab \equiv abcab$
\end{enumerate}
\end{theorem}

\begin{proof}
We only prove the first; the rest are similar.  From Theorem~\ref{kura} we get
$acacaca \equiv aca$.  Hence $ab(acacaca) \equiv ab(aca)$, or
equivalently, $aba(cacaca) \equiv aba(ca)$.
Since $aba \equiv ab$ from Lemma~\ref{kpk}, we get
$abcacaca \equiv abca$. \ \qed
\end{proof}



\section{Additional identities}

In this section we prove some additional
identities connecting the operations
$\lbrace k,e,c,p,s,f,w,r \rbrace$.

\begin{theorem}
We have
\begin{enumerate}[start=7]
\item $rp \equiv sr$
\item $rs \equiv pr$
\item $rf \equiv fr$
\item $rc \equiv cr$
\item $rk \equiv kr$
\item $rw \equiv wr$
\item $ps \equiv sp \equiv f$
\item $pf \equiv fp \equiv f$
\item $sf \equiv fs \equiv f$
\item $pw \equiv wp \equiv sw \equiv ws \equiv fw \equiv wf \equiv w$
\item $kw \equiv wk$
\item $rkw \equiv kw$
\item $ek \equiv ke \equiv k$
\item $fks \equiv pks$
\label{fks}
\item $fkp \equiv skp$ \label{fkp}
\item $rkf \equiv skf \equiv pkf \equiv fkf \equiv kf$ \label{pkf}
	\label{skf}
\end{enumerate}
\end{theorem}

\begin{proof}
All of these are relatively straightforward.  
To see (\ref{fks}), note that $p(L) \subseteq f(L)$ for all $L$, and
hence $pks(L) \subseteq fks(L)$.  Hence it suffices to show the reverse
inclusion.

Note that every element of $ks(L)$ is either $\epsilon$
or can be written $x = s_1 s_2 \cdots s_n$ for some $n \geq 1$, where each
$s_i \in s(L)$.  In the latter case,
any factor of $x$ must be of the form
$y = s''_i s_{i+1} \cdots s_{j-1} s'_j$, where
$s''_i $ is a suffix of $s_i$ and $s'_j$ is a prefix of $s_j$.
Then $s''_i s_{i+1} \cdots s_{j-1} s_j \in ks(L)$ and hence
$y \in pks(L)$. 

Similarly, we have
$pkf \equiv pk(ps) \equiv (pkp)s \equiv (kp)s \equiv k(ps) = kf$,
which proves part of (\ref{pkf}).
\end{proof}

\newpage
\begin{theorem}
We have
\begin{enumerate}[start=23]
\item 
$pcs(L) %
= \Sigma^*$ or $\emptyset$. \label{pcs} \label{first1}
\item The same result holds for 
$pcf, %
fcs, %
fcf, %
scp, %
scf, %
fcp, %
wcp, %
wcs, %
wcf, %
pcw, %
scw,$ %
\\
$fcw, %
wcw %
$.
	\label{pcf}
\end{enumerate}
\end{theorem}

\begin{proof}
Let us prove the first statement.
Either $s(L) = \Sigma^*$, or $s(L)$ omits some word $v$.  In the former case,
$cs(L) = \emptyset$, and so $pcs(L) = \emptyset$.  In the latter case, we have
$s(L)$ omits $v$, so $s(L)$ must also omit $\Sigma^* v$ (for otherwise, if
$xv \in f(L)$ for some $x$, then $v \in s(L)$).  So $\Sigma^* v \subseteq
cs(L)$.  Hence $pcs(L) = \Sigma^*$.

The remaining statements are proved similarly. \ \qed
\end{proof}

The following result was proved in \cite[Theorems 2 and 3]{BGS}.

\begin{lemma}\label{lem:bgs}
 We have $ecece\equiv cece$. 
\end{lemma}

\begin{theorem}
Let $L$ be any language. 
\begin{enumerate}[start=25]
\item We have $kckck(L)=ckck(L)\cup\{\epsilon\}$. \label{kckck}
\end{enumerate}
\end{theorem}

\begin{proof}
First, suppose  $\epsilon\in L$.
Then $e(L)=k(L)$ and $ce(L)=ck(L)$. 
Since $\epsilon\notin ck(L)$, we obtain $ece(L)=eck(L)=kck(L)-\{\epsilon\}$. 
Then, $cece(L)=ckck(L) \cup \{\epsilon\}$. 
So $ecece(L)=kckck(L)$. 
From Lemma~\ref{lem:bgs}, we deduce 
$kckck(L)=ecece(L)=cece(L)=ckck(L) \cup \{\epsilon\}$.

Second, suppose $\epsilon\notin L$.
Then $e(L)=k(L)-\{\epsilon\}$ and $ce(L)=ck(L)\cup \{\epsilon\}$. 
We obtain $ece(L)=kck(L)$ and $cece(L)=ckck(L)$. 
So $ecece(L)=eckck(L)=kckck(L)-\{\epsilon\}$. 
From Lemma~\ref{lem:bgs}, we deduce
$kckck(L)=ecece(L) \cup \{\epsilon\}= cece(L) \cup\{\epsilon\}=ckck(L)\cup \{\epsilon\}$.
\end{proof}

\begin{lemma}
Let $L$ be any language.
\begin{itemize}
\item[(a)] If $xy \in kp(L)$ then $x \in kp(L)$ and $y \in kf(L)$.
\item[(b)] If $xy \in ks(L)$ then $x \in kf(L)$ and $y \in ks(L)$.
\item[(c)] If $xy \in kf(L)$ then $x, y \in kf(L)$.
\item[(d)] If $xy \in kw(L)$, then $x, y \in kw(L)$.
\end{itemize}
\label{kp-lemma}
\end{lemma}

\begin{proof}
We prove only (b), with the others being proved similarly.
If $xy \in ks(L)$, then $x \in pks(L)$ and $y \in sks(L)$.
But $s \subseteq f$, so $pks \subseteq pkf$, and $pkf = kf$
by (\ref{pkf}).  Hence $x \in kf(L)$.  Similarly,
$sks \equiv ks$ by Lemma~\ref{kpk}, so $ y\in ks(L)$.  \ \qed
\end{proof}

\begin{lemma}\label{lem:inclusion}
We have $pcpckp\subseteq kp$.
\end{lemma}

\begin{proof}
Let $x\in pcpckp(L)$. Then there exists $y$ such that $xy\in cpckp(L)$. So $xy\notin pckp(L)$. 
Then, for all $z$, we have $xyz\notin ckp(L)$. Hence  $xyz\in kp(L)$. Thus $x\in pkp(L)=kp(L)$.
\end{proof}

\begin{theorem}
Let $b \in \lbrace p, s, f, w \rbrace$.  Then
\begin{enumerate}[start=26]
\item $kcb(L) = cb(L)\cup \lbrace \epsilon \rbrace$
\label{kcp} 
\item $kckb(L) = ckb(L) \cup \lbrace \epsilon \rbrace$ \label{kckp} 
\item  $kbcbckb(L)=bcbckb(L)\cup\{\epsilon\}$.
\end{enumerate}
\end{theorem}

\begin{proof}
We prove only three of these identities; the others can be proved similarly.

$kcp(L) = cp(L)\cup \lbrace \epsilon \rbrace$:
Assume $x \in kcp(L)$.  Either $x = \epsilon$ or we can
write $x = x_1 x_2 \cdots x_n$ for some $n \geq 1$, where each
$x_i \in cp(L)$.  Then each $x_i \not\in p(L)$.  In particular
$x_1 \not\in p(L)$.  Then $x_1 x_2 \cdots x_n \not\in p(L)$, because
if it were, then $x_1 \in p(L)$, a contradiction.  Hence
$x \in cp(L)$.  

$kckp(L) = ckp(L) \cup \lbrace \epsilon \rbrace$:
Assume $x \in kckp(L)$.  Either $x = \epsilon$ or we can~write $x=$ $x_1 x_2 \cdots x_n$ 
for some $n \geq 1$, where each
$x_i \in ckp(L)$.  Then each $x_i \not\in kp(L)$.~In particular
$x_1 \not\in kp(L)$.  Hence $x_1 (x_2 \cdots x_n) \not\in kp(L)$,
because if it~were,~then $x_1 \in kp(L)$ by Lemma~\ref{kp-lemma}, a
contradiction.  Hence $x \not\in kp(L)$, so $x \in ckp(L)$,~as~desired. 

$kpcpckp(L)=pcpckp(L)\cup\{\epsilon\}$: 
Assume $x\in kpcpckp(L)$. 
Either $x = \epsilon$ or 
we can write $x=x_1\cdots x_n$, where each $x_i\in pcpckp(L)$. 
In particular, there exists $y$ such that $x_ny \in cpckp(L)$;
that is, $x_ny\not \in pckp(L)$. 
Assume $x\not\in pcpckp(L)$. Then $xy\not\in cpckp(L)$, so $xy\in pckp(L)$.
Then there exists $z$ such that $xyz\in ckp(L)$; that is,  $xyz\not\in kp(L)$. 
But from Lemma~\ref{lem:inclusion}, we know that every $x_i$ is in  $kp(L)$. 
Further, since $x_ny\not \in pckp(L)$, we have $x_nyz \not \in ckp(L)$;
that is, $x_nyz \in kp(L)$. 
This shows that $xyz=x_1\cdots x_{n-1}(x_nyz)$ belongs to $kp(L)$,
a contradiction. \ \qed
\end{proof}

\begin{theorem}
We have
\begin{enumerate}[start=29]
\item $sckp(L) %
= \Sigma^*$ or $\emptyset$. \label{sckp}
\item The same result holds for 
$fckp, %
pcks, %
fcks, %
pckf, %
sckf, %
fckf, %
wckp, %
wcks, %
wckf, $\\ %
$wckw, %
pckw, %
sckw, %
fckw %
$. \label{fckw}
\end{enumerate}
\end{theorem}

\begin{proof}
To prove \eqref{sckp}, note that
either $kp(L) = \Sigma^*$, or $kp(L)$ omits some word $v$.  In the
former case, $ckp(L) = \emptyset$, and so $sckp (L) = \emptyset$.  In the
latter case, we have $kp(L)$ omits $v$, so $kp(L)$ must also omit
$v \Sigma^*$ (for otherwise, if $vx \in kp(L)$ for some $x$, then
$v \in kp(L)$ by Lemma~\ref{kp-lemma}, a contradiction).  Then $v \Sigma^*
\in ckp(L)$ and hence $sckp(L) = \Sigma^*$.

The other results can be proved similarly. \ \qed
\end{proof}

\begin{lemma}
Let $L$ be any language.
\begin{itemize}
\item[(a)] If $xy \in skp(L)$, then $x, y \in skp(L)$.
\item[(b)] If $xy \in pks(L)$, then $x, y \in pks(L)$.
\end{itemize}
\label{skp-lemma}
\end{lemma}

\begin{proof}
We prove only (a), with (b) being proved similarly.

If $xy \in skp(L)$, then $x \in pskp(L)$ and $y \in sskp(L)$.
But $pskp \equiv (ps)kp \equiv fkp \equiv skp$ by (\ref{fkp}).
So $x \in skp(L)$.  Also, $sskp = skp$, so $y \in skp(L)$. \ \qed
\end{proof}

\begin{theorem}
We have
\begin{enumerate}[start=31]
\item  $scskp(L) %
= \Sigma^*$ or $\emptyset$.
\item The same result holds for $pcpks$.
\label{last1}
\end{enumerate}
\end{theorem}

\begin{proof}
We prove only the first result; the second can be proved analogously.
Either $skp(L) = \Sigma^*$, or it omits some word $v$.  In the 
first case we have $cskp(L) = \emptyset$ and hence
$scskp(L) = \emptyset$.  
In the second case, $skp(L)$ must omit 
$v \Sigma^*$ (for if $vx \in skp(L)$ for any $x$, then by Lemma~\ref{skp-lemma}
we have $v \in skp(L)$, a contradiction).  Hence $scskp(L) = \Sigma^*$. 
\ \qed
\end{proof}

\section{Results}

Our main result is the following:

\begin{theorem}
Let $S = \lbrace  k,e,c,p,f,s,w,r \rbrace$.  Then
for every language $L$, the set ${\cal O}_S (L)$ contains at most 5676 distinct
languages.
\label{main}
\end{theorem}

\begin{proof}
Our proof was carried out mechanically.  We used breadth-first search
to examine the set $S^* = \lbrace  k,e,c,p,f,s,w,r \rbrace^*$ by increasing
length of the words; within each length we used lexicographic order
with $k < e < c < p < f < s < w < r$.  The nodes remaining to be
examined are stored in a queue $Q$.

As each new word $x$ representing a series of language
operations is examined, we test it to see
if any factor is of the form given in identities
\eqref{first1}--\eqref{pcf} or \eqref{fckw}--\eqref{last1}.
If it is, then the corresponding
language must be either $\Sigma^*$, $\emptyset$, $\lbrace \epsilon
\rbrace$, or $\Sigma^+$; furthermore, each descendant language will be of this
form.  In this case the word $x$ is discarded.  

Otherwise, we use the remaining identities above to try to reduce $x$
to an equivalent word that we have previously encountered.
If we succeed, then $x$ is discarded.  Otherwise $x(L)$ is
potentially a new language, so we append all the words $Sx$
to the end of the queue.  Some simplifications are possible. For example,
using our identities we can assume $x$ contains only a single $r$ and this
appears at the end; this cuts down on the search space.

We treat the identities \eqref{kckck}--\eqref{kckp} somewhat differently.
We keep track of whether a language contains $\epsilon$ or not.
For example, when appropriate,
we can replace $akcb$ with $acb$ 
for $a,b\in \{ p, s, f, w \}$.  

If the process terminates, then ${\cal O}_S (L)$ is of finite cardinality.

We wrote our program in APL.  For 
$S = \lbrace  k,c,p,f,s,w,r \rbrace$, the process terminated with
5672 nodes that could not be simplified using our identities.  
We did not count $\emptyset, \lbrace \epsilon \rbrace,
\Sigma^+,$ and $\Sigma^*$.  The total is thus 5676.

The longest word examined was $ckcpcpckpckpckpcpcpckckcr$, of length 25,
and the same word with $p$ replaced by $s$.  

Our program generates a complete description of the words and how
they simplify, which can be viewed at
{\tt www.cs.uwaterloo.ca/\char'176shallit/papers.html}. \ \qed
\end{proof}

\begin{remark}
If we use {\it two\/} arbitrary
closure operations $a$ and $b$ with no relation 
between them, then the monoid generated by
$\lbrace a,b \rbrace$ could potentially be
infinite, since any two finite prefixes of $ababab\cdots$ are distinct.

Here is an example.  Let $p$ denote prefix, as above, and define
the exponentiation operation
\begin{equation}
t(L) = \lbrace x^i \ : \ x \in L \text{ and } i \text{ is an integer} \geq 1 \rbrace.
\label{expo}
\end{equation}
Then it is easy to see that $t$ is a closure operation, and hence
the orbits ${\cal O}_{\lbrace p \rbrace} (L)$ 
and ${\cal O}_{\lbrace t \rbrace} (L)$ are finite, for
all $L$.  However, for $L = \lbrace ab \rbrace$,
the orbit ${\cal O}_{\lbrace p,t \rbrace} (L)$ is infinite,
as $a b a^i \in (pt)^i (L)$, but
$a b a^i \not\in (pt)^j (L)$ for all $j < i$.

Thus our proof of Theorem~\ref{main} 
crucially depends on the properties of the operations
$\lbrace k,e,c,p,s,f,w,r \rbrace$.
\end{remark}

We now give some results for some interesting subsets of $S$.

\subsection{Prefix and complement}

In this case at most 14 distinct languages can be generated.  The bound of $14$
can be achieved, e.g., by the regular language over $\Sigma = \lbrace a,b,c,d \rbrace$
given by the regular expression
$ a^*( (b+c)(a (\Sigma\Sigma)^* + b + d \Sigma^*) + d \Sigma^+)$
and accepted by the DFA in Figure~\ref{pc-fig}. 
\begin{figure}[H]
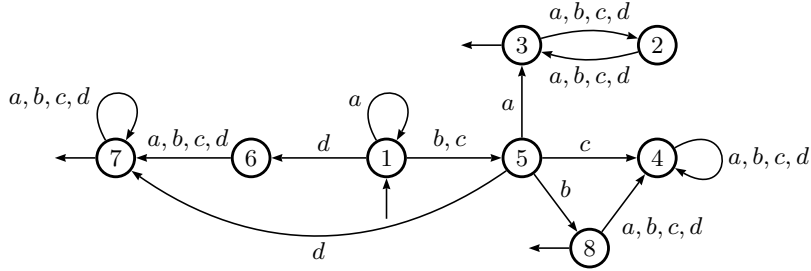

\centering
\VCDraw{%
\begin{VCPicture}{(-6,-2.5)(6,2.5)}
 \State[1]{(0,0)}{1}
 \State[5]{(3,0)}{5}
 \State[4]{(6,0)}{4}
 \State[3]{(3,2.5)}{3}
 \State[2]{(6,2.5)}{2}
 \State[6]{(-3,0)}{6}
 \State[7]{(-6,0)}{7}
\State[8]{(4.5,-2)}{8}
\Initial[s]{1}
\Final[w]{7}
\Final[w]{3}
\Final[w]{8}
 \LoopN{1}{a}
\EdgeR{1}{6}{d}
\EdgeL{1}{5}{b,c}

\ArcL[.5]{3}{2}{a,b,c,d}
\ArcL[.5]{2}{3}{a,b,c,d}
\VArcL[.5]{arcangle=35}{5}{7}{d}
\EdgeL{5}{4}{c}
\EdgeL{5}{3}{a}
\EdgeL{5}{8}{b}
\EdgeR{6}{7}{a,b,c,d}
\EdgeR[.3]{8}{4}{a,b,c,d}
\LoopE[.5]{4}{a,b,c,d}
\LoopN{7}{a,b,c,d}
\end{VCPicture}}
    \caption{DFA accepting a language $L$ with orbit size $14$ under operations $p$ and $c$}
\label{pc-fig}
\end{figure}

Table~\ref{table1} gives the appropriate set of final states under the operations.
\begin{table}[H]
\begin{center}
\begin{tabular}{|c|c||c|c|}
\hline
language & final states & language & final states  \\
\hline
$L$ & 3,7,8 & $pcpc(L)$ & 1,5,6,7 \\
\hline
$c(L)$ & 1,2,4,5,6 & $cpcp(L)$ & 2,3,6,7 \\
\hline
$p(L)$ & 1,2,3,5,6,7,8 & $cpcpc(L)$ & 2,3,4,8 \\
\hline
$pc(L)$ & 1,2,3,4,5,6,8 & $pcpcp(L)$ & 1,2,3,5,6,7 \\
\hline
$cp(L)$ & 4 & $pcpcpc(L)$ & 1,2,3,4,5,8 \\
\hline
$cpc(L)$ & 7 & $cpcpcp(L)$ & 4, 8 \\
\hline
$pcp(L)$ & 1,4,5,8 & $cpcpcpc(L)$ & 6, 7  \\
\hline
\end{tabular}
\end{center}
\caption{Final states for composed operations}
\label{table1}
\end{table}

\subsection{Prefix, Kleene star, complement}

The same process, described above for the operations
$\lbrace k, e, c, p, s, f, w, r \rbrace$, can be carried out for
other subsets, such as $\lbrace k, c, p \rbrace$.  For this our breadth-first
search gives $1066$ languages. The longest word examined
was $ckcpcpckpckpckpcpcpckckc$.

\subsection{Factor, Kleene star, complement}

Similarly, we can examine $\lbrace k, c, f \rbrace$.  
Here breadth-first search gives $78$ languages, so our bound is $78+ 4 = 82$.
We can improve this bound by considering new kinds of arguments. 

\begin{lemma}\label{lem:kcf1} 
Let $L$ be any language. 
There are at most 4 languages distinct from $\Sigma^*, \emptyset,\Sigma^+$, and $\{\epsilon\}$ in 
$\mathcal O_{\{k,f,kc,fc\}}(f(L))$. These languages are among $f(L),$ $kf(L),$ $kckf(L)$, and $kcf(L)$.
\end{lemma}

\begin{proof}[Sketch]
First observe that the set of languages $\{\Sigma^*, \emptyset,\Sigma^+,\{\epsilon\}\}$
is closed under any operation of the set $\{k,c,f\}$. We make a case study.
We consider successively the languages generated by $\{k,f,kc,fc\}$ from $fcf(L),kf(L)$, and $kcf(L)$. 
We make use of Identities \eqref{skf}, \eqref{pcf}, \eqref{kcp}, \eqref{kckp}, and \eqref{fckw}.
\end{proof}

Let $\alphabet(L)$ denote the minimal alphabet of a language $L$, that is, the minimal set of letters that occur in words of $L$.

\begin{lemma}\label{lem:kf} 
Let $L$ be any language. We have $kf(L)=k(\alphabet(L))$.
\end{lemma}

\begin{proof}
The minimal alphabets of $L$ and $f(L)$ coincide. 
Thus $f(L)\subseteq k(\alphabet(L))$, so $kf(L)\subseteq k(\alphabet(L))$. 
Further, $\alphabet(L)\subseteq f(L)$. 
So $k(\alphabet(L))\subseteq kf(L)$~as~well.
\end{proof}

\begin{lemma}\label{lem:kcf2} 
Let $L$ be any language.
There are at most 2 languages distinct from $\Sigma^*, \emptyset,\Sigma^+$, and $\{\epsilon\}$ in 
$\mathcal O_{\{k,f,kc,fc\}}(fk(L))-\mathcal O_{\{k,f,kc,fc\}}(f(L))$. These languages are among $fk(L)$ and $kcfk(L)$.
\end{lemma}

\begin{proof}
Apply Lemma~\ref{lem:kcf1} to $k(L)$ and use $kfk\equiv kf$. To see the latter identity, use Lemma~\ref{lem:kf} 
and observe that $\alphabet(k(L))=\alphabet(L)$.
\end{proof}

\begin{lemma}  \label{f-fc}
For any language $L$, we have either $f(L)=\Sigma^*$ or $fc(L)=\Sigma^*$.
\end{lemma}

\begin{proof} 
Assume $f(L)\neq\Sigma^*$. 
Then there exists a word in $cf(L)$, say $w$. Hence $\Sigma^*w\Sigma^*\cap f(L)=\emptyset$. 
Since $L\subseteq f(L)$, we also have  $\Sigma^*w\Sigma^*\cap L=\emptyset$, that is $\Sigma^*w\Sigma^*\subseteq c(L)$. This implies $fc(L)= \Sigma^*$. 
\end{proof}

\begin{theorem}
50 is a tight upper bound for the size of the orbit of $\{k,c,f\}$.
\end{theorem}

\begin{proof}[Sketch]
From Lemmas~\ref{lem:kcf1} and \ref{lem:kcf2}, and Identity~\eqref{kckck},
starting with an arbitrary language $L$, 
the languages in $\mathcal O_{\{k,c,f\}}(L)$ that may differ from $\Sigma^*, \emptyset,\Sigma^+$, and $\{\epsilon\}$ 
are among the images of $L$ and $c(L)$ under the 16 operations
\begin{gather}\label{list16}
f,kf,kckf,kcf,fk,kcfk,fck,kfck,kckfck,kcfck,\\
\nonumber fkck,kcfkck,fckck,kfckck,kckfckck,kcfckck.
\end{gather}
the complements of these images, together with the $14$ languages in $\mathcal O_{\{k,c\}}(L)$.

By using Lemma \ref{f-fc}, we show that there are at most 32 pairwise distinct languages among the $64=16\cdot 4$ 
languages given by the images of $L$ and $c(L)$ under the 16 operations \eqref{list16} and their complements.

Adding the 14 languages in $\mathcal O_{\{k,c\}}(L)$, and $\Sigma^*, \emptyset,\Sigma^+$, and $\{\epsilon\}$, we obtain that $50=32+14+4$ is an upper bound 
for the size of the orbit of $\{k,c,f\}$.

The bound is tight because the language $L$ given by two copies (over disjoint alphabets) of the language  accepted by the DFA of Figure~\ref{fig:aut50} over the alphabet $\{a,b,c,d,e,f,g,h,i\}$ (i is a letter that does not occur in any word of $L$, i.e., $i\notin \alphabet(L)$) has 50 pairwise distinct elements in $\mathcal O_{\{k,c,f\}}(L)$. 
\begin{figure}[H]
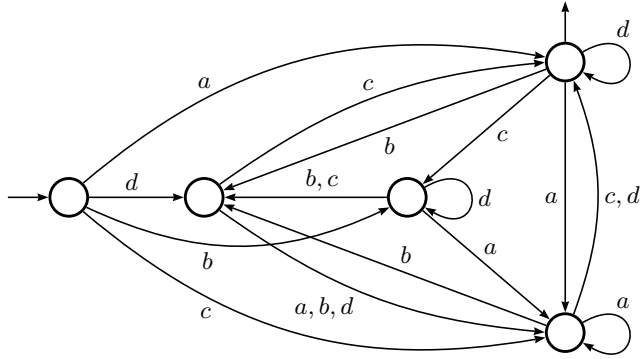

\centering
\VCDraw{%
\begin{VCPicture}{(0,-4)(12,3)}
 \State{(0,0)}{1}
 \State{(3,0)}{5}
 \State{(7.5,0)}{3}
 \State{(11,3)}{2}
 \State{(11,-3)}{4}
\Initial[w]{1}
\Final[n]{2}
\EdgeL{1}{5}{d}
\VArcL[.3]{arcangle=25}{1}{2}{a}
\VArcR[.4]{arcangle=-25}{1}{3}{b}
\VArcR[.3]{arcangle=-25}{1}{4}{c}

\LoopE{2}{d}
\EdgeR[.5]{2}{4}{a}

\EdgeL[.5]{2}{5}{b}
\EdgeL{2}{3}{c}

\LoopE[.5]{3}{d}
\EdgeL{3}{4}{a}
\EdgeR[.4]{3}{5}{b,c}

\LoopE{4}{a}
\VArcR[.5]{arcangle=-20}{4}{2}{c,d}
\EdgeR[.45]{4}{5}{b}

\VArcL[.5]{arcangle=15}{5}{2}{c}
\VArcR[.4]{arcangle=-15}{5}{4}{a,b,d}

\end{VCPicture}}
    \caption{The DFA made of two copies of this DFA accept a language $L$ with orbit size $50$ under operations $k,c,$ and $f$.}
\label{fig:aut50}
\end{figure}
\end{proof}

\subsection{Kleene star, prefix, suffix, factor}

Here there are at most $13$ distinct languages, given by the action of
$$\lbrace \epsilon, k, p, s, f, kp, ks, kf, pk, sk, fk, pks, skp \rbrace.$$
The bound of $13$ is achieved, for example, by $L = \lbrace abc \rbrace$.

\subsection{Summary of results}

Table~\ref{table2} gives our upper bounds on the number of distinct
languages generated by the set of operations.  An entry in {\bf bold}
indicates that the bound is known to be tight.  Some entries, such as
$p,r$, are omitted, since they are the same as others (in this case,
$p,s,f,r$).  Most bounds were obtained directly from our program, and others
by additional reasoning.  An asterisk denotes those bounds for which some
additional reasoning was required to reduce the upper bound found by our
program to the bound shown in Table~\ref{table2}.

\begin{table}[htb]
\begin{center}
\begin{tabular}{|M|N||M|N||M|N|}
\hline
  r & {\bf 2} &             w & {\bf 2} &          f & {\bf 2} \\
\hline
  s & {\bf 2} &             p & {\bf 2} &           c & {\bf 2} \\
\hline
  k & {\bf 2} &             w,r & {\bf 4} &          f,r & {\bf 4} \\
\hline
  f,w & {\bf 3} &            s,w & {\bf 3} &          s,f & {\bf 3} \\
\hline
  p,w & {\bf 3} &            p,f & {\bf 3} &          c,r & {\bf 4} \\
\hline
  c,w & {\bf 6}* &            c,f & {\bf 6}* &         c,s & {\bf 14} \\
\hline
  c,p & {\bf 14} &           k,r & {\bf 4} &          k,w & {\bf 4} \\
\hline
  k,f & {\bf 5} &            k,s & {\bf 5} &          k,p & {\bf 5} \\
\hline
  k,c & {\bf 14} &           f,w,r & {\bf 6} &         s,f,w & {\bf 4} \\
\hline
  p,f,w & {\bf 4} &           p,s,f & {\bf 4} &        c,w,r & {\bf 10}* \\
\hline
  c,f,r & {\bf 10}* &          c,f,w & {\bf 8}* &        c,s,w & {\bf 16}* \\
\hline
  c,s,f & {\bf 16}* &          	c,p,w & {\bf 16}* &        c,p,f & {\bf 16}* \\
\hline
 k,w,r & {\bf 7} &          k,f,r & {\bf 9} &        k,f,w & {\bf 6} \\
\hline
  k,s,w & {\bf 7} &           k,s,f & {\bf 9} &        k,p,w & {\bf 7} \\
\hline
  k,p,f & {\bf 9} &           k,c,r & {\bf 28} &        k,c,w &  {\bf 38}* \\
\hline
  k,c,f & {\bf 50}*  &         k,c,s & 1070 &      k,c,p & 1070 \\
\hline
  p,s,f,r & {\bf 8} &          p,s,f,w & {\bf 5} &        c,f,w,r & {\bf 12}* \\
\hline
  c,s,f,w & {\bf 16}* &         c,p,f,w & {\bf 16}* &       c,p,s,f & {\bf 16}* \\
\hline
  k,f,w,r & {\bf 11} &         k,s,f,w & {\bf 10} &       k,p,f,w & {\bf 10} \\
\hline
  k,p,s,f & {\bf 13} &         k,c,w,r & 72* &      k,c,f,r & {\bf 84}* \\
\hline
  k,c,f,w & 66* &        k,c,s,w & 1114 &     k,c,s,f & 1450 \\
\hline
  k,c,p,w & 1114 &       k,c,p,f & 1450 &     p,s,f,w,r & {\bf 10} \\
\hline
  c,p,s,f,r & {\bf 30}* &        c,p,s,f,w & {\bf 16}* &      k,p,s,f,r & {\bf 25} \\
\hline
  k,p,s,f,w & {\bf 14} &        k,c,f,w,r & 120* &     k,c,s,f,w & 1474 \\
\hline
  k,c,p,f,w & 1474 &     k,c,p,s,f & 2818 &    c,p,s,f,w,r & {\bf 30}*  \\
\hline
  k,p,s,f,w,r & {\bf 27} &       k,c,p,s,f,r & 5628 &   k,c,p,s,f,w & 2842 \\
\hline
  k,c,p,s,f,w,r & 5676  & &&& \\
\hline
\end{tabular}
\end{center}
\caption{Upper bounds on the size of the orbit}
\label{table2}
\end{table}

\section{Further work}

We plan to continue to refine our estimates in Table~\ref{table2}, and
pursue the status of other sets of operations.  For example, if
$t$ is the exponentiation operation defined in (\ref{expo}), then,
using the identities $kt = tk = k$, and the inclusion $t \subseteq k$,
we get the additional Kuratowski-style identities
$kctckck \equiv kck$, 
$kckctck \equiv kck$,
$kctctck \equiv kck$,
$tctctck \equiv tck$, and
$kctctct \equiv kct$.
This allows us to prove that  ${\cal O}_{\lbrace k,c,t \rbrace} (L)$ is
finite and of cardinality at most $126$.

\section{Acknowledgments}

We thank John Brzozowski for his comments.


\begin{thebibliography}{1}

\bibitem{BGS}
Janusz Brzozowski, Elyot Grant, and Jeffrey Shallit.
\newblock Closures in formal languages and {K}uratowski's theorem.
\newblock In {\em Developments in Language Theory}, volume 5583 of {\em Lecture
  Notes in Comput. Sci.}, pages 125--144. Springer, Berlin, 2009.

\bibitem{GJ}
B.~J. Gardner and M.~Jackson.
\newblock The {K}uratowski closure-complement theorem.
\newblock {\em New Zealand J. Math.}, 38:9--44, 2008.

\bibitem{Kuratowski:1922}
C.~Kuratowski.
\newblock Sur l'op{\'e}ration {$\overline{A}$} de l'analysis situs.
\newblock {\em Fund. Math.}, 3:182--199, 1922.

\bibitem{Peleg}
D.~Peleg.
\newblock A generalized closure and complement phenomenon.
\newblock {\em Discrete Math.}, 50(2-3):285--293, 1984.

\end{thebibliography}
\end{document}